\documentclass[11pt,a4paper]{article}

\usepackage[a4paper, total={5.8in, 9in}]{geometry}

\usepackage{amsmath,amsfonts}
\usepackage{amsthm}
\usepackage{amssymb,hyperref}
\usepackage[english]{babel}
\usepackage{enumitem}
\usepackage{mdframed}

\usepackage{latexsym}
\usepackage{epsfig}
\usepackage{bbm}
\usepackage{thm-restate}

\usepackage{esvect}

\usepackage{xcolor} 
\hypersetup{
    colorlinks,
    linkcolor={blue},
    citecolor={blue},
    urlcolor={blue}
}

\newtheorem{theorem}{Theorem}
\newtheorem{lemma}[theorem]{Lemma}

\begin{document} \bibliographystyle{alpha}
\title{Complex Semidefinite Programming and {\sc Max-$k$-Cut}}

\author{\textsc{Alantha
    Newman}\thanks{CNRS and Universit\'e Grenoble Alpes.  
Supported in part by LabEx
    PERSYVAL-Lab (ANR-11-LABX-0025).}}

\maketitle

\def\mkc{{\sc Max}-$k$-{\sc Cut}}
\def\mkcs{{\sc Max}-$k$-{\sc Cut} }

\begin{abstract}
In a second seminal paper on the application of semidefinite
programming to graph partitioning problems, Goemans and Williamson
showed how to formulate and round a {\em complex semidefinite program}
to give what is to date still the best-known approximation guarantee
of .836008 for {\sc Max-$3$-Cut}.  (This approximation
ratio was also achieved independently by De Klerk et
al.)  Goemans and Williamson left
open the problem of how to apply their techniques to {\sc Max-$k$-Cut}
for general $k$.  They point out that it does not seem straightforward
or even possible to formulate a good quality complex semidefinite
program for the general {\sc Max-$k$-Cut} problem, which presents a
barrier for the further application of their techniques.

We present a simple rounding algorithm for the standard semidefinite
programmming relaxation of {\sc Max-$k$-Cut} and show that it is
equivalent to the rounding of Goemans and Williamson in the case of
{\sc Max-$3$-Cut}.  This allows us to transfer the elegant analysis of
Goemans and Williamson for {\sc Max-3-Cut} to {\sc Max-$k$-Cut}.  For
$k \geq 4$, the resulting approximation ratios are about $.01$ worse
than the best known guarantees.  Finally, we present a generalization
of our rounding algorithm and conjecture (based on computational
observations) that it matches the best-known guarantees of De Klerk et
al.
\end{abstract}

\section{Introduction}

In the {\sc Max-$k$-Cut} problem, we are given an undirected graph,
$G=(V,E)$, with non-negative edge weights.  Our objective is to divide the
vertices into at most $k$ disjoint sets, for some given positive
integer $k$, so as to maximize the weight of the edges whose endpoints
lie in different sets.  When $k=2$, this problem is known simply as
the {\sc Max-Cut} problem.  The approximation guarantee of $1-1/k$
can be achieved for all $k$ by placing each vertex uniformly at
random in one of $k$ sets.  For all values of $k \geq 2$, this
simple algorithm yielded the best-known approximation ratio until
1994.  In that year, Goemans and Williamson gave a
.87856-approximation algorithm for the {\sc Max-Cut} problem based on semidefinite
programming (SDP), thereby introducing this method as a successful new
technique for designing approximation algorithms~\cite{GW}.  

Frieze and Jerrum subsequently developed an algorithm for the {\sc
  Max-$k$-Cut} problem that can be viewed as a generalization of
Goemans and Williamson's algorithm for {\sc Max-Cut} in the sense that
it is same algorithm when $k=2$~\cite{FJ}.  Although the rounding
algorithm of Frieze and Jerrum is arguably simple and natural, the
analysis is quite involved.  Their approximation ratios improved upon
the previously best-known guarantees of $1-1/k$ for $k \geq 3$ and are
shown in Table \ref{tbl:chart}.  A few years later, Andersson,
Engebretsen and H\aa stad also used semidefinite programming to design
an algorithm for the more general problem of {\sc Max-E2-Lin mod $k$},
in which the input is a set of equations or inequations mod $k$ on two
variables (e.g., $x-y \equiv c ~\bmod k$) and the objective is to
assign an integer from the range $[0,k-1]$ to each variable so that
the maximum number of equations are satisfied~\cite{AEH}.  They proved
that the approximation guarantee of their algorithm is at least $f(k)$
more than that of the simple randomized algorithm, where $f(k)$ is a
(small) linear function of $k$.  In the special case of {\sc
  Max-$k$-Cut}, they showed that the performance ratio of their
algorithm is no better than that of Frieze and Jerrum.  Although they
did not show the equivalence of these two algorithms, they stated that
numerical evidence suggested that the two algorithms have the same
approximation ratio.  Shortly thereafter, De Klerk, Pasechnik and
Warners presented an algorithm for {\sc Max-$k$-Cut} with improved
approximation guarantees for all $k\geq3$, shown in Table
\ref{tbl:chart}.  Additionally, they showed that their algorithm has
the same worst-case performance guarantee as that of Frieze and
Jerrum~\cite{DBLP:journals/jco/KlerkPW04}.

Around the same time, Goemans and Williamson independently presented
another algorithm for {\sc Max-3-Cut} based on {\it complex
  semidefinite programming} (CSDP)~\cite{GW2}.  For this problem, they
improved the best-known approximation guarantee of $.832718$ due to
Frieze and Jerrum to $.836008$, the same approximation ratio proven by
De Klerk, Pasechnik and Warners.  Goemans and Williamson showed that
their algorithm is equivalent to that of Andersson, Engebretsen and
H\aa stad and to that of Frieze and Jerrum (and therefore to that of
De Klerk, Pasechnik and Warners) in the case of {\sc
  Max-3-Cut}~\cite{GW2}.  However, they argued that their decision to
use complex semidefinite programming and, specifically, their choice
to represent each vertex by a single complex vector resulted in
``cleaner models, algorithms, and analysis than the equivalent models
using standard semidefinite programming.''

One issue noted by Goemans and Williamson with respect to their elegant
new model was that it is not clear how to apply their techniques to
{\sc Max-$k$-Cut} for $k \geq 4$.  Their approach
seemed to be tailored specifically to the {\sc Max-3-Cut} problem.  This is
because one cannot model, say, the {\sc Max-$4$-Cut} problem directly
using a complex semidefinite program.  This limitation is discussed in Section 8
of \cite{GW2}.  In fact, as they point out, a direct attempt to model
{\sc Max-$k$-Cut} with a complex semidefinite program would only
result in a $(1-1/k)$-approximation for $k\geq 4$.
De Klerk et al. also state that there is no obvious way
to extend the approach based on CSDP to {\sc Max-$k$-Cut} for $k >
3$.  (See page 269 in \cite{DBLP:journals/jco/KlerkPW04}.)

\subsection{Our Contribution}

In this paper, we make the following contributions.

\begin{enumerate}

\item {We present a simple rounding algorithm based on the standard
  semidefinite programming relaxation of {\sc Max-$k$-Cut} and show that
  it can be analyzed using the tools from \cite{GW2}.  

\begin{itemize}

\item For $k=3$, this results in an implementation of the Goemans-Williamson
  algorithm that avoids complex semidefinite
  programming.

\item For $k \geq 4$,
  the resulting approximation ratios are slightly worse than the
  best-known guarantees.

\end{itemize}
}

\item We present a simple generalization of this rounding algorithm and
  conjecture that it yields the best-known approximation ratios.

\end{enumerate}

Thus, the main contribution of this paper is to show that, despite its limited
modeling power, we can still apply the tools from complex semidefinite
programming developed by Goemans and Williamson to {\sc Max-$k$-Cut}.
In fact, we obtain the following worst-case approximation guarantee
for the {\sc Max-$k$-Cut} problem for all $k$, which is the same bound
they achieve for $k=3$:
\begin{eqnarray}
\phi_k & = & \frac{k-1}{k} 
+ \frac{k}{4\pi^2} \left[\arccos^2\left(\left(\frac{1}{k-1}\right)\cos\left( \frac{2\pi}{k} \right) \right)
 - \arccos^2\left(\frac{1}{k-1}\right) \right].
\end{eqnarray}
We note that for $k \geq 4$, the approximation ratio $\phi_k$ is about
$.01$ worse than the approximation ratio proved by Frieze and Jerrum.
See Table \ref{tbl:chart} for a comparison.  
However, given the
technical difficulty of Frieze and Jerrum's analysis, we believe that
it is beneficial to present an alternative algorithm and analysis that
yields a similiar approximation guarantee.  Moreover, we wish to take
a closer look at the techniques used by Goemans and Williamson for
{\sc Max-3-Cut} since these tools have not been widely applied in the
area of approximation algorithms, in sharp contrast to the tools used
to solve the {\sc Max-Cut} problem.  In fact, we are aware of only two
papers that use the main tools of \cite{GW2}: The first is for a
generalization of the {\sc Max-3-Cut}
problem~\cite{ling2009approximation} and the second is for an
optimization problem in which the variables are to be assigned complex
vectors~\cite{zhang2006complex}.

While Goemans and Williamsons' framework of complex semidefinite
programming does result in an elegant formulation and analysis for
{\sc Max-3-Cut}, it also to some extent obscures the geometric
structure that is apparent when one views the same algorithm from the
viewpoint of standard semidefinite programming.  Specifically, in the
latter framework, their complex semidefinite program is equivalent to
modeling each vertex with a 2-dimensional circle or disc of vectors.
In our opinion, their main technical contribution is a formula for the
exact distribution of the difference of the angles resulting when a
normal vector is projected onto two of these discs that are correlated
in a particular way. (See Lemma 8 in \cite{GW2}.)  Thus, while the
limitation in modeling {\sc Max-$k$-Cut} with complex semidefinite
programming comes from the fact that we cannot model the general
problem with these 2-dimensional discs, we can circumvent this barrier
in the following way.  We construct 2-dimensional discs using the
vectors obtained from a solution to the standard semidefinite program.
We then show that a pair of these 2-dimensional discs (i.e., one disc
for each vertex) are correlated in the same way as those produced in
the case of {\sc Max-3-Cut}.  Then we can apply and analyze the same
algorithm used for {\sc Max-3-Cut}.

In some cases (e.g., {\sc Max-3-Cut}), using the
distribution of the angle between two elements is stronger than using
the expected angle, which is what is used for {\sc Max-Cut}.  It
therefore seems that this tool has unexplored potential applications
for other optimization problems, for which it may also be possible to
overcome the modeling limitations of complex semidefinite programming
in a similiar manner as we do here.  On a high level, the idea of
constructing the ``complex'' vectors from a solution to a standard
semidefinite program was used for a circular arrangement
problem~\cite{DBLP:conf/innovations/MakarychevN11}.

Finally, we remark that the approach used in Section \ref{sec:mkcut}
to create a disc from a vector is reminiscent of Zwick's method of
outward rotations in which he combines hyperplane rounding and
independent random assignment~\cite{Zwick}.  For each unit vector
$v_i$ from an SDP solution, he computes a disc in the plane spanned by
$v_i$ and $u_i$, where the $u_i$'s form a set of pairwise orthogonal
vectors that are also orthogonal to the $v_i$'s, and chooses a new
vector from this disc based on a predetermined angle.  Thus, the goal
is to rotate each vector $v_i$ to obtain a new set of unit vectors, which are
then given as input to a now standard rounding algorithm, such as
random-hyperplane rounding.  In contrast, our goal is to use the actual disc in
the rounding, as done originally by Goemans and Williamson in the case
of {\sc Max-3-Cut}.

\begin{table}
\begin{center}
\fbox{\parbox{11cm}{
\begin{tabular}
{c|c|c|c|c|c}
$k$ & \cite{GW} & \cite{FJ} & \cite{GW2} & \cite{DBLP:journals/jco/KlerkPW04} & This paper\\
$k=2$ & .878956 & - & - & - & - \\
$k=3$ & -  & .832718 & .836008 & .836008 & -\\
$k=4$ & - & .850304 & - & .857487 & .846478\\
$k=5$ & - & .874243 & - & .876610 & .862440\\
$k=10$ & - & .926642 & - & .926788 & .915885
\end{tabular}}}
\caption{Approximation guarantees for {\sc Max-$k$-Cut}.}\label{tbl:chart}
\end{center}
\end{table}

\subsection{Organization}

We give some background on the (standard) semidefinite programming
relaxation used by Frieze and Jerrum and discuss their algorithm for
{\sc Max-$k$-Cut} in Section \ref{sec:FJ}.  In Section \ref{sec:GW},
we present Goemans and Williamson's algorithm for {\sc Max-3-Cut} from
the viewpoint of standard semidefinite programming.  In Section
\ref{sec:mkcut}, we show how to create a 2-dimensional disc for each
vertex given a solution to the standard semidefinite program for {\sc
  Max-$k$-Cut}.  We do not wish to formally prove the relationship
between these discs and the complex vectors.  Thus, in Section
\ref{sec:analysis}, we simply prove that if two discs are correlated
in a specified way, then the distribution of the angle is equivalent
to a distribution already computed exactly by Goemans and Williamson
in \cite{GW2}.  We can then easily prove that the 2-dimensional discs
we create for the vertices have the required pairwise correlation.
This results in a closed form approximation ratio for general $k$,
Theorem \ref{thm:main}.

\section{Frieze and Jerrum's Algorithm}\label{sec:FJ}

Consider
the following integer program for {\mkc}: 
\begin{eqnarray*}
& \max & \sum_{ij \in E} (1 - v_i \cdot v_j) \frac{k-1}{k} \\
v_i \cdot v_i & = & 1, \quad \forall i \in V,\\
v_i & \in & \Sigma_k, \quad \forall i \in V. \hspace{20mm} (P)
\end{eqnarray*}
Here, $\Sigma_k$ are the vertices of the equilateral simplex, where
each vertex is represented by a $k$-dimensional vector, and each pair
of vectors corresponding to a pair of vertices has dot product $-1/(k-1)$.  If
we relax the dimension of the vectors, we obtain the following
semidefinite relaxation, where $n = |V|$:
\begin{eqnarray*} 
& \max & \sum_{ij \in E} (1 - v_i \cdot v_j) \frac{k-1}{k}\nonumber \\
v_i \cdot v_i & = & 1, \quad \forall i \in V,\nonumber\\
v_i\cdot v_j & \geq & -\frac{1}{k-1}, \quad \forall i, j \in V,\nonumber\\
v_i & \in & \mathbb{R}^n, \quad \forall i \in V. \hspace{20mm} (Q)
\end{eqnarray*}
Frieze and Jerrum used this semidefinite relaxation to obtain an
algorithm for the {\sc Max}-$k$-{\sc Cut} problem~\cite{FJ}.
Specifically, they proposed the following rounding algorithm: Choose
$k$ random vectors, $g_1, g_2, \dots, g_k \in \mathbb{R}^n$, with each entry of each
vector chosen from the normal distribution ${\cal{N}}(0,1)$.  For
each vertex $i\in V$, consider the $k$ dot products of vector $v_i$
with each of the $k$ random vectors, $v_i\cdot g_1, v_i \cdot g_2,
\dots, v_i \cdot g_k$.  One of these dot products is maximum.  Assign
the vertex the label of the random vector with which it has the
maximum dot product.  In other words, if $v_i \cdot g_h = \max_{\ell=1}^k \{v_i
\cdot g_{\ell}\}$, then vertex $i$ is assigned to to cluster $h$.
Frieze and Jerrum were able to prove a lower bound on the
approximation guarantee of this algorithm for every $k$.  See Table
\ref{tbl:chart} for some of these ratios.  

\section{Goemans-Williamson Algorithm for {\sc Max-3-Cut}}\label{sec:GW}

Goemans and Williamson gave an algorithm for {\sc Max-3-Cut} in which
they first model the problem as a complex semidefinite program (i.e.,
each element is represented by a complex vector).  It is not too
difficult to see that these complex vectors are equivalent to
2-dimensional discs or sets of unit vectors.  For example, here is an
equivalent semidefinite program for {\sc Max-3-Cut}.  The input is an
undirected graph $G=(V,E)$ with non-negative edge weights
$\{w_{ij}\}$.
\begin{eqnarray}
& \max & \sum_{ij\in E} w_{ij} (1-v_i^1 \cdot v_j^1)\frac{2}{3}\\
v_i^a \cdot v_i^b & = & -1/2, \hspace{13mm} \forall i \in V, ~ a \neq b \in [3],\\
v_i^a \cdot v_j^b & = & v_i^{a+c} \cdot v_j^{b+c}, \hspace{5mm} \forall i,j
\in V, ~ a,b,c \in [3],\\
v_i^a \cdot v_j^b & \geq & -1/2, \hspace{13mm}\forall i,j \in V, ~a,b \in [3],\\
v_i^a \cdot v_i^a & = & 1, \hspace{20mm} \forall i \in V, ~ a \in [3],\\
v_i^a & \in & \mathbb{R}^{3n}, \hspace{15mm} \forall i \in V, ~a \in [3].
\end{eqnarray}
Consider a set of $3n$ unit vectors forming a solution to this
semidefinite program.
Note that for a fixed vertex $i \in V$, the vectors $v_i^1, v_i^2$ and
$v_i^3$ are in the same 2-dimensional plane, since they are
constrained to be pairwise $120^{\circ}$ apart.  
\begin{figure}[t]
\begin{center}
\epsfig{file = 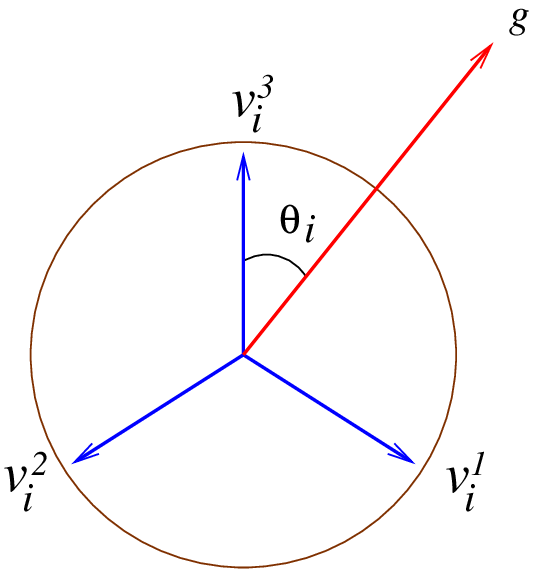, width=3.2cm}\hspace{10mm}
\epsfig{file = 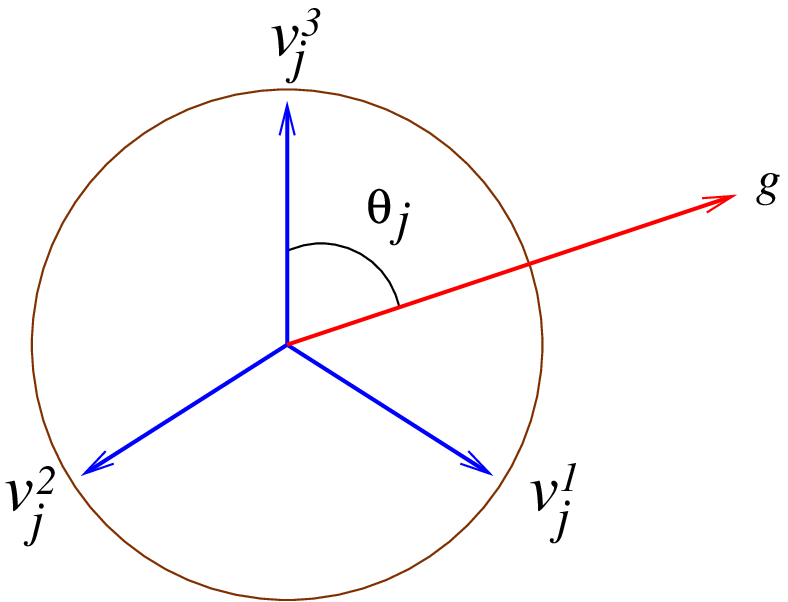, width=4cm}\\
\epsfig{file = 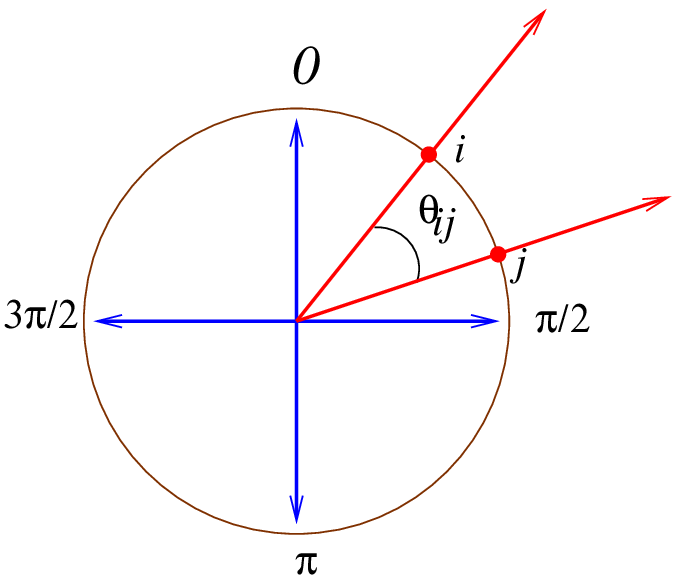, width=4cm}\hspace{6mm}
\epsfig{file = 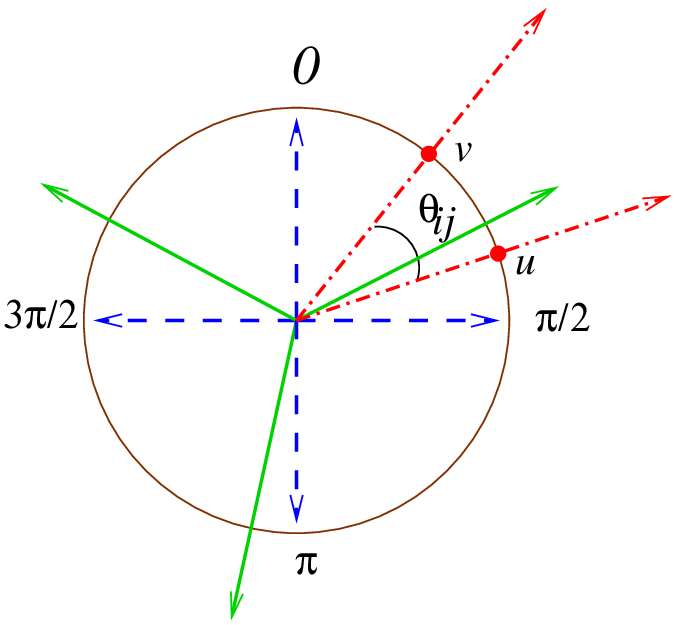, width=4cm}
  {\caption{Three vectors $v_i^1, v_i^2$ and $v_i^3$ lie on a
      2-dimensional plane corresponding to vertex $i$.  The vector $g$
      is projected onto the disc for element $i$ to obtain $\theta_i$.  Angle
      $\theta_{ij}$ is the difference between angles $\theta_i$ and $\theta_j$.} \label{fig:circle}}
\end{center}
\end{figure}
In an ``integer'' solution for this semidefinite program, all these discs would be
constrained to be in the same 2-dimensional space and each angle of
rotation of the discs would be constrained to be $0, 2\pi/3$ or
$4\pi/3$, where each angle would correspond to a partition.  In a
solution to the above relaxation, these discs are no longer
constrained to be in two dimensions.

In the rounding algorithm of Goemans and Williamson, we first pick a
vector $g \in \mathbb{R}^{3n}$ such that each entry is chosen
according to the normal distribution ${\cal N}(0,1)$.  Then for each
vertex $i \in V$, we project this vector $g$ onto its corresponding
disc.  This gives an angle $\theta_i$ in the range $[0,2\pi)$ for
  each element $i$.  (Note that without loss of generality, we can
  assume that $\theta_i$ is the angle in the clockwise direction
  between the projection of $g$ and the vector $v_i^3$.)  We can
  envision the angles $\{\theta_i\}$ for each $i \in V$ embedded onto
  the same disc.  Then we randomly partition this disc into three
  equal pieces, each of length $2\pi/3$ (i.e., we choose an angle $\psi
  \in [0, 2\pi)$ and let the three angles of partition be $\psi, \psi
    + 2\pi/3$ and $\psi + 4\pi/3$).  These three pieces correspond to
    the three sets in the partition.

The angle $\theta_{ij}$ is the angle $\theta_j - \theta_j$ modulo
$2\pi$.  The probability that an edge $ij$ is cut in this
partitioning scheme is equal to $3\theta_{ij}/2\pi$ if $\theta_{ij} <
2\pi/3$ and 1 otherwise.  In expectation, the angle $\theta_{ij}$ is
equal to $\arccos{(v_i^1 \cdot v_j^1)}$.  (This can be shown using the
techniques in \cite{GW}.  See Lemma 3 in
\cite{DBLP:conf/innovations/MakarychevN11}.)  But using the expected
angle is not sufficient to obtain an approximation guarantee better
than $2/3$; If angle $\theta_{ij}$ is $2\pi/3$ in expectation, then
one third of the time it could be zero (not cut) and two thirds of the
time it could be $\pi$ (cut).  However, it contributes 1 to the
objective function.
The exact probability that edge $ij$ is cut is:
\begin{eqnarray*}
\Pr[\text{edge } ij \text{ is cut}] ~ = ~ \sum^{2\pi/3}_{\gamma = 0} \Pr[\theta_{ij} = \theta] 
  \times \frac{\theta}{2\pi/3} + \sum^{4\pi/3}_{\gamma = 2\pi/3}
  \Pr[\theta_{ij} = \theta] + \sum_{\theta = 4\pi/3}^{2\pi}
  \Pr[\theta_{ij} = \theta] \times \frac{2\pi-\theta}{2\pi/3}.
\end{eqnarray*}
Therefore, we must compute $\Pr[\theta_{ij} = \theta]$ for all $\theta
\in [0,2\pi)$.  One of the main technical contributions of Goemans and
  Williamson~\cite{GW2} is that they compute the exact probability
  that $\theta_{ij} < \delta$ for all $\delta \in [0, 2\pi)$.  This
    can be found in Lemma 8~\cite{GW2}.  This enables them to compute
    the probability that an edge is cut, resulting in their
    approximation guarantee.

\section{Algorithm for {\sc Max-$k$-Cut}}\label{sec:mkcut}

As previously mentioned, we cannot model {\sc Max-$k$-Cut} as an
integer program directly using 2-dimensional discs as we do for {\sc
  Max-3-Cut}, because any rotation corresponding to an angle of at
least $2\pi/k$ should contribute 1 to the objective function.  Note
that in the case of {\sc Max-3-Cut}, there are two possible non-zero
rotations in an integer solution: $2\pi/3$ and $4\pi/3$ and both of
the contribute the same amount (i.e., 1) to the objective function.
Since it seems impossible to penalize all angles greater than $2\pi/k$
at the same cost, it seems similiarly impossible to model the problem
directly with a complex semidefinite program.

We now present our approach for rounding the semidefinite programming
relaxation $(Q)$ for {\mkc}.  
After solving the semidefinite program, we obtain a set of vectors
$\{v_i\}$ corresponding to each vertex $i \in V$.  We can assume these
vectors to be in dimension $n$.  Let ${\bf{0}}$ represent the vector
with $n$ zeros.  For each vertex $i \in V$, we construct the following
two orthogonal vectors:
\begin{eqnarray}
v_i ~ :=  ~(v_i,{\bf 0}), \quad \quad 
v_i^{\perp} ~ := ~ ({\bf 0}, v_i).
\end{eqnarray}  
Each vertex $i \in V$ now corresponds to a 2-dimensional disc spanned
by vectors $v_i$ and $v_i^{\perp}$.  Specifically, this 2-dimensional
disc consists of the
(continuous) set of vectors defined for $\phi \in [0,2\pi)$:
\begin{eqnarray}\label{def:phi}
v_i(\phi) & = & v_i \cos{\phi} + v_i^{\perp}\sin{\phi}.
\end{eqnarray}
Now that we have constructed a 2-dimensional disc for each element,
we can use the same rounding scheme due to Goemans and Williamson
described in the previous section:  First, we choose a vector
$g \in \mathbb{R}^{2n}$ in
which each coordinate is randomly chosen according to the normal
distribution ${\cal N}(0,1)$.  For each $i \in V$, we project this
vector $g$ onto the disc $\{v_i(\phi)\}$, which results in an angle
$\theta_i$, where:
\begin{eqnarray*}
g \cdot v_i(\theta_i) & = & \max_{0 \leq \phi < 2\pi} g \cdot v_i(\phi).
\end{eqnarray*}
Note that we do not have to compute infinitely many dot products,
since, for example, if $g\cdot v_i, ~g \cdot v_i^{\perp} \geq 0$, then:
\begin{eqnarray*}
\theta_i & = & \arctan{\left(\frac{g \cdot v_i^{\perp}}{g\cdot v_i}\right)},
\end{eqnarray*}
and the three other cases depending on the sign of $g\cdot v_i$ and $g \cdot
v_i^{\perp}$ can be handled accordingly.

After we find an angle $\theta_i$ for each $i \in V$, we can assign
each element to a position corresponding to its angle $\theta_i$ on a
single disc and divide this disc (randomly) into $k$ equal sections
of size $2\pi/k$.  Specifically, choose a random angle $\psi$ and use
the partition $\psi + \frac{c \cdot 2\pi}{k}$ for all integers $c \in
[0,k)$, where angles are taken modulo $2\pi$.  These are the $k$
  partitions of the vertices in the $k$-cut.

\begin{figure}[t]
\begin{center}
\epsfig{file = 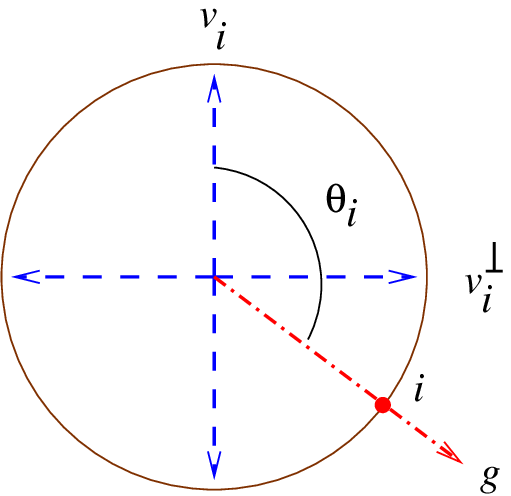, width=4cm}\hspace{3mm}
\epsfig{file = 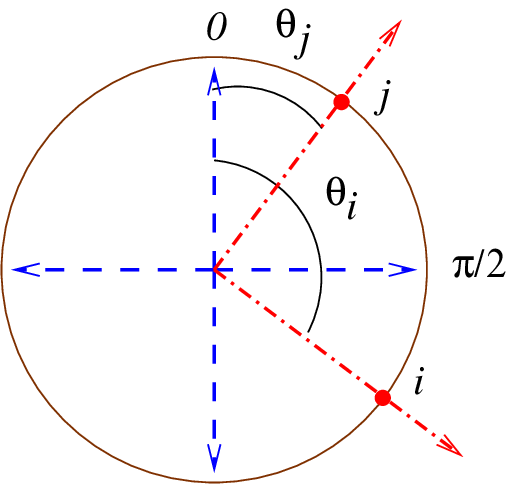, width=4cm}
\epsfig{file = 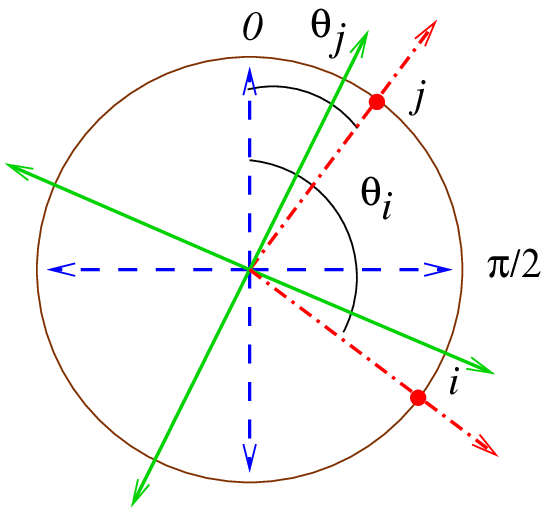, width=4cm}
  {\caption{A 2-dimensional plane for vertex $i$ spanning $v_i$ and
      $v_i^{\perp}$.  After projecting $\theta_i$ and $\theta_j$ onto
      the same disc, we partition the disc into $k=4$ equal sized
      pieces.}}
\end{center}
\end{figure}

\section{Analysis}\label{sec:analysis}

We prove that the distribution of the angle $\theta_{ij}$ is the
same as Lemma 8 of \cite{GW2}.  This implies that we can use the 
analysis that Goemans and Williamson use for {\sc Max-3-Cut} to obtain an analogous
approximation ratio for {\sc Max-$k$-Cut}.  
\begin{lemma}
Given two sets of vectors $x_i = \{x_i(\phi)\}$ and $x_j =
\{x_j(\phi)\}$ defined on $\phi \in [0, 2\pi)$, where
\begin{eqnarray*}
x_i(\phi) & = & (\cos\phi, ~\sin\phi,  ~0, ~0),\\
x_j(\phi) & = & (\cos\theta\cos\phi, ~\cos\theta\sin\phi,
~\sin\theta\cos\phi, ~\sin\theta\sin\phi).
\end{eqnarray*}
Let $\gamma \in [0, 2\pi)$ denote the angle $\theta_j - \theta_i$ after the
vector $g \in {\cal N}(0,1)^{2n}$ is projected onto $x_i$ and $x_j$.
Then for $\delta \in [0,2\pi)$,
\begin{eqnarray}
\Pr[0 \leq \gamma < \delta] = \frac{1}{2\pi} \left(\delta + \frac{r \sin{\delta}}{\sqrt{1-r^2\cos^2{\delta}}}
\arccos{(-r \cos{\delta})} \right).\label{form:last}
\end{eqnarray}
\end{lemma}

\begin{proof}
Note that the set of vectors $x_j$ is 2-dimensional, since
the 
angle between $x_j(\phi_1)$ and $x_j(\phi_2)$ for $\phi_2 > \phi_1$ is
$\phi_2-\phi_1$.  Thus, the rounding algorithm in Section
\ref{sec:mkcut} is well defined.
Recall that each coordinate of the vector $g$
is chosen according to the normal distribution ${\cal N}(0,1)$.
Even though the vector $g$ has $2n$ dimensions, we only need to
consider the first four, $g = (g_1, g_2, g_3, g_4)$.
This vector is chosen
equivalently to choosing $\alpha,\beta$ uniformly in $[0,2\pi)$ and
  $p_1, p_2$ according to the distribution:
$$f(y) = y e^{-y^2/2}.$$ 
In other words, the vector $g$ is equivalent to:
\begin{eqnarray*}
g & = & (p_1\cos{\beta}, ~p_1 \sin{\beta}, ~p_2\cos{\alpha}, ~p_2\sin{\alpha}).
\end{eqnarray*}
Let $r = \cos{\theta}$ and let $s = \sin{\theta}$.  
We will show that the probability that $\gamma \in [0,\delta)$ for
  $\delta \leq \pi$ is:
\begin{eqnarray}
\Pr[0 \leq \gamma < \delta] & = & \frac{1}{2\pi} \left[ \delta +
  \int^{\pi}_{\delta} \Pr\left[\frac{p_2 \cdot s}{\sin{\delta}}
    \leq \frac{p_1 \cdot r}{\sin{(\alpha-\delta)}}\right]
  d\alpha\right]. \label{fourteen}
\end{eqnarray}
Lemma 8 in \cite{GW2} shows this is equivalent to probability in
\eqref{form:last}.  

First, let us consider the case when $\theta \in [0,\pi/2]$, or
$\cos\theta \geq 0$.
Without loss of generality, assume that the projection of $g$ onto the
2-dimensional disc $x_i$ occurs at $\phi = 0$.  Then we can see that
\begin{eqnarray*}
x_i(0) \cdot g = p_1.
\end{eqnarray*}
In other words, we can assume that $\theta_i = 0$.
As previously mentioned, $\alpha$ is chosen uniformly in the range
$[0,2\pi)$.  However, if $\gamma < \delta$, then $\alpha < \pi$.
If $\alpha < \delta$, then the projection of $g$ onto
  $x_j$, namely $\theta_j$ (which equals $\theta_{ij}$ in this case,
  because we have assumed that $\theta_i = 0$), is less than $\delta$.
The probability
that $\gamma \leq \delta$ if $\alpha \in [\delta, \pi)$ is equal to
the probability
that:
\begin{eqnarray*}
\frac{p_2 \cdot s}{\sin{\delta}} & \leq & \frac{p_1 \cdot
  r}{\sin{(\alpha-\delta)}} \quad \iff\\
p_2 \cdot s & \leq & \frac{p_1 \cdot
  r}{\sin{(\alpha-\delta)}} \cdot \sin{\delta}.
\end{eqnarray*}
(See Figure 3 in \cite{GW2}.)
If $\theta \in (\pi/2, \pi)$ and $r = \cos\theta < 0$, then the
probability that $\gamma$ is in $[0, \delta)$ is the probability that
  $\gamma$ is in $[\pi, \pi+\delta)$, which is $\delta/(2\pi)$.  And
    the probability that $\gamma$ is in $[\delta, \pi)$ is the
      probability that $\gamma$ is in $[\pi+\delta, 2\pi)$ for $-r$.
        This is:
\begin{eqnarray}
p_2 \cdot s & \leq & \frac{p_1 \cdot
  (-r)}{\sin{(\alpha-\delta)}} \cdot \sin{(\pi + \delta)}.
\end{eqnarray}
However, since $\sin{(\pi + \delta)} = - \sin{\delta}$, we have:
\begin{eqnarray}
p_2 \cdot s & \leq & \frac{p_1 \cdot
  r}{\sin{(\alpha-\delta)}} \cdot \sin{\delta}.
\end{eqnarray}
Thus for all $\delta < \pi$, we have proved the expression in
\eqref{fourteen}.
In Lemma 8 of \cite{GW2}, they show that Equation \eqref{fourteen} 
is equivalent to Equation \eqref{form:last} when $\delta < \pi$.  Then
they argue by symmetry that Equation \eqref{form:last} also holds when
$\pi \leq \delta < 2\pi$.\end{proof}

\begin{lemma}
Suppose $v_i \cdot v_j = \cos\theta$ for two unit vectors $v_i$ and
$v_j$.  Let $v_i(\phi)$ and $v_j(\phi)$ be defined as in equation
\eqref{def:phi}.  Then, we can assume that:
\begin{eqnarray*}
v_i(\phi) & = & (\cos\phi, ~\sin\phi,  ~0, ~0),\\
v_j(\phi) & = & (\cos\theta\cos\phi, ~\cos\theta\sin\phi,
~\sin\theta\cos\phi, ~\sin\theta\sin\phi).
\end{eqnarray*}
\end{lemma}

\begin{proof}
From the definition (in Equation \eqref{def:phi}) of $v_i(\phi)$, we
can see that:
\begin{eqnarray*}
v_i(\phi_1) \cdot v_j(\phi_2) & = & (v_i \cos{\phi_1} +
v_i^{\perp}\sin{\phi_1})\cdot (v_j \cos{\phi_2} + v_j^{\perp}
\sin{\phi_2})\\
& = & v_i \cdot v_j \cos\phi_1 \cos\phi_2 + v_i^{\perp} \cdot
v_j^{\perp} \sin\phi_1 \sin\phi_2 + v_i \cdot v_j^{\perp} \cos\phi_1
\sin\phi_2 + v_i^{\perp} \cdot v_j \sin\phi_1 \cos\phi_2\\
& = & 
\cos\theta
(\cos{\phi_1}\cos{\phi_2} + \sin{\phi_1}\sin{\phi_2}).
\end{eqnarray*}
Note that $v_i \cdot v_j^{\perp} = v_i^{\perp}\cdot v_j = 0$ since
each $v_i$ vector has $n$ zeros in the second half of the entries and each
$v_i^{\perp}$ vector has $n$ zeros in the first half of the entries.
If we compute $v_i(\phi_1)\cdot v_j(\phi_2)$ using the assumption in
the lemma, then we get the same dot product.  Thus, the two sets are
equivalent.\end{proof}

\vspace{5mm}

Since the distribution of the angle is the same, we can use the
same analysis of \cite{GW2} (generalized from $3$ to $k$) to prove the
following Lemma.  Although it is essentially the exact same proof, we
include it here for completeness.  
As in Corollary 9 of \cite{GW2}, we define:
\begin{eqnarray*}
g(r,\delta) & = & \frac{1}{2\pi}\left(\delta + \frac{r \sin{\delta}}{\sqrt{1-r^2\cos^2{\delta}}}\arccos{(-r \cos{\delta})} \right).
\end{eqnarray*}
In other words, $g(r,\delta)$ is the probability that angle
$\theta_{ij}$ obtained by projecting $g$ onto
the two discs $\{v_i{(\phi)}\}$ and $\{v_j(\phi)\}$, correlated
by $r = v_i \cdot v_j$, is less than $\delta$.

\begin{lemma}\label{closed_form}
Let $r = v_i\cdot v_j$ and let $y_i \in [0,1,2, \dots k)$ be the
  integer assignment of vertex $i$ to its partition.  Then the probability that the equation $y_i -
y_j \equiv c ~(\bmod ~k)$ is satisfied is
\begin{eqnarray*}
 \frac{1}{k} + \frac{k}{8\pi^2}\left[2
    \arccos^2\left(-r\cos\left(\frac{2\pi c}{k}\right)
    \right) - \arccos^2\left(-r\cos\left( \frac{2\pi(c+1)}{k} \right) \right) -
    \arccos^2\left(-r \cos\left( \frac{2\pi(c-1)}{k}\right)
    \right) \right].
\end{eqnarray*}
\end{lemma}

\begin{proof}
$\Pr[y_i-y_j \equiv c ~(\bmod ~k) \text{
    satisfied}]$
\begin{align*}
& =  \frac{k}{2\pi} \int^{\frac{k}{2\pi}}_0
 \Pr_{\gamma}\left[\frac{2\pi c}{k}-\tau \leq \gamma <
   \frac{2\pi(c+1)}{k} -\tau \right] d\tau\\
& =  \frac{k}{2\pi} \int^{\frac{2\pi}{k}}_0 \left( g\left(\tau, \frac{2\pi(c+1)}{k}
 - \tau) - g(\tau, \frac{2\pi c}{k} - \tau\right)
\right) d\tau\\
& =  \frac{k}{2\pi} \int^{\frac{2\pi(c+1)}{k}}_{\frac{2\pi
    c}{k}} g(r, \nu)d\nu - \frac{k}{2\pi} \int^{\frac{2\pi
    c}{k}}_{\frac{2\pi(c-1)}{k}} g(r,\nu) d\nu\\
& =  \frac{k}{2\pi}\frac{1}{2\pi}
\left( \int^{\frac{2\pi(c+1)}{k}}_{\frac{2\pi c}{k}}
\nu d\nu -\left[\frac{1}{2}\arccos^2{(-r\cos{\nu})}
  \right]^{\frac{2\pi(c+1)}{k}}_{\frac{2\pi
    c}{k}} \nonumber \right.\\ 
& \quad \quad \left.- \int^{\frac{2\pi c}{k}}_{\frac{2\pi(c-1)}{k}} \nu
d\nu + \left[\frac{1}{2} \arccos^2{(-r\cos{\nu})}
  \right]^{\frac{2\pi c}{k}}_{\frac{2\pi(c-1)}{k}} \right) \\
& =  \frac{k}{8\pi^2} \left[ \left( \frac{2\pi(c+1)}{k}\right)^2 + \left(\frac{2\pi(c-1)}{k}\right)^2 -2
  \left(\frac{2\pi c}{k}\right)^2 \nonumber \right]\\
& \quad \quad + \frac{k}{8\pi^2}\left[2 \arccos^2\left(-r \cos \left(\frac{2\pi c}{k}\right) \right) \right. \nonumber\\
& \quad \quad \left.- \arccos^2\left(-r\cos \left(\frac{2\pi(c+1)}{k} \right)\right)
-\arccos^2\left(-r \cos \left(\frac{2\pi(c-1)}{k}\right)
  \right) \right]\\
& =  \frac{1}{k} + \frac{k}{8\pi^2}\left[2
    \arccos^2\left(-r\cos\left(\frac{2\pi c}{k}\right)
    \right) \right. \nonumber\\
& \quad \quad \left. - \arccos^2\left(-r\cos\left( \frac{2\pi(c+1)}{k} \right) \right) -
    \arccos^2\left(-r \cos\left( \frac{2\pi(c-1)}{k}\right)
    \right) \right].
\end{align*}
\end{proof}

\begin{lemma}\label{lemm:not}
Let $r = v_i\cdot v_j$.
The
probability that edge $ij$ is {\em not} cut by our algorithm is:
\begin{eqnarray*}
\frac{1}{k} + \frac{k}{4\pi^2} \left[\arccos^2\left(-r\right)-
  \arccos^2\left(-r\cos\left( \frac{2\pi}{k} \right) \right)\right].
\end{eqnarray*}
\end{lemma}

\begin{proof}
In the case of {\sc Max-$k$-Cut}, we set $c=0$.  By Lemma
\ref{closed_form}, we have the probability that edge $ij$ is not
cut is:
\begin{eqnarray*}
\frac{1}{k} + \frac{k}{8\pi^2} \left[2 \arccos^2\left(-r\right)- \arccos^2\left(-r\cos\left( \frac{2\pi}{k} \right) \right) -
    \arccos^2\left( -r \cos\left( -\frac{2\pi}{k} \right)
    \right)\right]
\end{eqnarray*}
\begin{eqnarray*}
& = & \frac{1}{k} + \frac{k}{8\pi^2} \left[2 \arccos^2\left(-r\right)-
  2\arccos^2\left(-r\cos\left( \frac{2\pi}{k} \right) \right)\right]\\
& = & \frac{1}{k} + \frac{k}{4\pi^2} \left[\arccos^2\left(-r\right)-
  \arccos^2\left(-r\cos\left( \frac{2\pi}{k} \right) \right)\right].
\end{eqnarray*}\end{proof}

\begin{lemma}
Let $r = v_i \cdot v_j$.  The probability that
edge $ij$ is cut by our algorithm is:
\begin{eqnarray}
\frac{k-1}{k} 
+ \frac{k}{4\pi^2} \left[\arccos^2\left(-r\cdot \cos\left( \frac{2\pi}{k} \right) \right)
 - \arccos^2{(-r)} \right].\label{edge_guar}
\end{eqnarray}
\end{lemma}

\begin{proof}
By Lemma \ref{lemm:not} and the previously stated assumption that $r = v_i
\cdot v_j = \cos{(\theta_{ij})}$, we have:
\begin{eqnarray*}
&& 1 - \left[\frac{1}{k} + \frac{k}{4\pi^2} \left[\arccos^2\left(-r\right)-
  \arccos^2\left(-r\cos\left( \frac{2\pi}{k} \right) \right)\right]
  \right] \\
& = & \frac{k-1}{k} - \frac{k}{4\pi^2} \left[\arccos^2\left(-r\right)-
  \arccos^2\left(-r\cos\left( \frac{2\pi}{k} \right) \right)\right]\\
& = & \frac{k-1}{k} 
+ \frac{k}{4\pi^2} \left[\arccos^2\left(-r\cos\left( \frac{2\pi}{k} \right) \right)
 - \arccos^2\left(-r\right) \right].
\end{eqnarray*}
\end{proof}

\begin{theorem}\label{thm:main}
The worst case approximation ratio of our algorithm for {\sc
  Max-$k$-Cut} is:
\begin{eqnarray*}
\phi_k & = & \frac{k-1}{k} 
+ \frac{k}{4\pi^2} \left[\arccos^2\left(\left(\frac{1}{k-1}\right)\cos\left( \frac{2\pi}{k} \right) \right)
 - \arccos^2\left(\frac{1}{k-1}\right) \right].
\end{eqnarray*}
\end{theorem}

\begin{proof}
As a function of $r$ in the range $[1,-1/(k-1)]$, the
expression in Equation \ref{edge_guar} is minimized when $r =
-1/(k-1)$.  Thus, if we do an edge-by-edge analysis, the worst case
approximation ratio is obtained when $v_i \cdot v_j = -1/(k-1)$ for
all edges $ij \in E$.  \end{proof}

\section{Another Rounding Algorithm}\label{sec:another}

The algorithm presented in Section 4 can be restated as the following
rounding scheme.  Let $w_1, w_2$ and $w_3$ denote vectors in
$\mathbb{R}^2$ with pairwise dot product $-1/2$.  In other words,
$w_1, w_2$ and $w_3$ are the vertices of the simplex $\Sigma_3$.  Now
take two random gaussians $g_1, g_2 \in \mathbb{R}^n$ and set $x_i =
g_1 \cdot v_i$, $y_i = g_2 \cdot v_i$.  To assign the vertex $i$ to
one of the three partitions, we simply assign it to $j$ such that $w_j
\cdot (x_i, y_i)$ is maximized.

We can generalize this approach by choosing $k-1$ random gaussians,
$g_1, \dots, g_{k-1}$.  For each vertex $i$, we obtain the vector
$(g_1\cdot v_i, g_2 \cdot v_i, \dots, g_{k-1} \cdot v_i)$ in
$\mathbb{R}^{k-1}$.  This vector is assigned to the closest vertex of
$\Sigma_k$.  Computationally, this rounding scheme seems to yield 
approximation ratios that match those of De Klerk et al.

\section{Acknowledgements}

Thanks to Moses Charikar, Anupam Gupta, R. Ravi and
Madhur Tulsiani for
helpful discussions and comments on the presentation.

\bibliography{kcut}

\end{document}